%% file: Steiner.tex
\documentclass[11pt]{article}
\usepackage{amsmath}
\usepackage{algorithm}
\usepackage{algpseudocode}
\usepackage[margin=1in]{geometry}
\usepackage{amsmath, amsthm, amssymb}
\usepackage{hyperref}
\usepackage{authblk}
\usepackage{graphicx}
\usepackage{mathrsfs}
\usepackage{enumitem, verbatim}
\usepackage{xcolor}
\usepackage{tikz}
\usetikzlibrary{3d}
\include{macros-2}

\usepackage{tikz}
\usepackage{tikz-3dplot}
\usetikzlibrary{3d,calc}

\title{Approximating mixed volumes to arbitrary accuracy}

\author{Hariharan Narayanan}
\author{Sourav Roy}

\affil{School of Technology and Computer Science, Tata Institute of Fundamental Research}

\date{}

\begin{document}

\maketitle

\begin{abstract}
We study the problem of approximating the mixed volume $V(P_1^{(\a_1)}, \dots, P_k^{(\a_k)})$ of an $k$-tuple of convex polytopes $(P_1, \dots, P_k)$, each of which is defined as the convex hull of at most $m_0$ points in $\Z^n$.  We design an algorithm that produces an estimate that is within a multiplicative $1 \pm \epsilon$ factor of the true mixed volume with a probability greater than $1 - \de.$  Let the constant $ \prod_{i=2}^{k} \frac{(\alpha_{i}+1)^{\alpha_{i}+1}}{\alpha_{i}^{\,\alpha_{i}}}$ be denoted by $\tilde{A}$.  When each $P_i \subseteq B_\infty(2^L)$, we show in this paper that the time complexity of the algorithm is bounded above by a polynomial in $n, m_0, L, \tilde{A}, \epsilon^{-1}$ and $\log \de^{-1}$. In fact, a stronger result is proved in this paper, with slightly more involved terminology.

In particular, we provide the first randomized polynomial time algorithm for computing mixed volumes of such polytopes when  $k$ is an absolute constant, but $\a_1, \dots, \a_k$ are arbitrary. 
Our approach synthesizes tools from convex optimization, the theory of Lorentzian polynomials, and  polytope subdivision. 
 \end{abstract}

\section{Introduction}
Mixed volumes are central objects in convex geometry, capturing how the volume of a Minkowski sum of convex bodies depends on the shapes and sizes of the summands. They play a foundational role in the Brunn–Minkowski theory, and the Alexandrov–Fenchel inequalities.
Given an $k$-tuple of convex compact sets \(K = (K_1, \ldots, K_k)\) where for each $i$, \(K_i \subset \mathbb{R}^n \), the Minkowski polynomial is defined as
    \[
        V_{K}(\lambda_1, \ldots, \lambda_k) = \text{Vol}_n(\lambda_1 K_1 + \cdots + \lambda_k K_k),
    \]
    and the mixed volume is
    \[
        V(K_1^{(\a_1)}, \ldots, K_k^{(\a_k)}) = \left(\frac{1}{n!}\right) \left.\frac{\partial^n}{(\partial \lambda_1)^{\a_1} \cdots (\partial \lambda_k)^{\a_k}} V_K(\lambda_1, \ldots, \lambda_k) \right.
    \]

Mixed volumes also play a role in numerical algebraic geometry, through their appearance in the Bernstein–Kushnirenko–Khovanskii (BKK) theorem \cite{BKK,  kush, khov}, which connects convex geometry with the theory of sparse polynomial systems. The BKK theorem asserts that the number of solutions in $(\mathbb{C}^\times)^n$ of a generic system of $n$ polynomials in $n$ variables is equal to  the mixed volume of their Newton polytopes. 

\begin{theorem}[BKK]
Let $f_1, \dots, f_n$ be Laurent polynomials in $n$ variables of the form
\[
f_i(x) = \sum_{a \in A_i} c_{i,a} x^a, \quad x = (x_1, \dots, x_n), \quad x^a = x_1^{a_1} \cdots x_n^{a_n},
\]
where $A_i \subset \mathbb{Z}^n$ is finite, and $c_{i,a} \in \mathbb{C}$.

Let $P_i = \operatorname{conv}(A_i)$ be the Newton polytope of $f_i$, for $i = 1, \dots, n$.

Then the number of isolated solutions in $(\mathbb{C}^\times)^n$ to the system
\[
f_1(x) = f_2(x) = \dots = f_n(x) = 0
\]
counted with multiplicities, is at most $n!$ times the mixed volume i.e. $n! \operatorname{V}(P_1, \dots, P_n)$.

Moreover, for generic choices (i.e. in  the complement of a set of Lebesgue measure $0$) of the coefficients $c_{i,a}$, the number of isolated solutions in $(\mathbb{C}^\times)^n$ is exactly equal to $n!$ times the the mixed volume:
\[
\#\left\{x \in (\mathbb{C}^\times)^n : f_1(x) = \dots = f_n(x) = 0\right\} = n! \operatorname{V}(P_1, \dots, P_n).
\]
\end{theorem}

The work of Dyer--Gritzmann--Hufnagel \cite{DGH} operates in the oracle model for what are defined to be well-presented convex bodies. This class includes polytopes. Their paper gives both complexity results and some positive approximation results. The authors show that computing mixed volumes is \#P-hard in general, even when the volume is easy to compute and all bodies are axis-aligned boxes or zonotopes. They provide a randomized polynomial-time algorithm for approximating mixed volumes when 
\[
m_1 \geq n - \psi(n) \quad \text{with} \quad \psi(n) = o\left( \frac{\log n}{\log \log n} \right),
\]
where $m_1$ is the multiplicity of the first convex body. 

\begin{remark}
On the one hand, work of Dyer--Gritzmann-Hufnagel \cite{DGH} considers convex sets that are not necessarily polytopes, and thus encompasses bodies for which Theorem~\ref{thm:main} above is not applicable. On the other,  it is also more restrictive in that when $k$ is a constant, (and the convex bodies are contained in $B_\infty(2^L)$ the convex hulls of polynomially many lattice points) not all mixed volumes are shown to be approximable within $ 1 \pm \eps$  in randomized polynomal time. These mixed volumes (as well as some more general ones) can be computed in polynomial time using Theorem~\ref{thm:main}.
\end{remark}

We work with rational numbers in this paper and account for their full bit complexity.
   \subsection{Capacity:} When $k = n$ and all the $\a_i = 1$, Gurvits \cite{Gurvits} defined the capacity 
    \[
        \text{Cap}(V) = \inf_{\mathbf{x} > 0} \frac{V_K(x_1, \ldots, x_n)}{\prod_{i=1}^n x_i},
    \]
    which is used as a  proxy for \( V(K_1, \ldots, K_n) \).

 Gurvits,  in \cite{Gurvits},  also presented a randomized polynomial-time algorithm that approximates the mixed volume \( V(K_1, \ldots, K_n) \) of an \(n\)-tuple of convex compact subsets \(K_1, \ldots, K_n \subset \mathbb{R}^n\) within a multiplicative factor of asymptotically \( e^n \) by showing the following. 
    \[
        \frac{n!}{n^n} \cdot \text{Cap}(V) \leq n! V(K_1, \ldots, K_n) \leq \text{Cap}(V).
    \]
    
    \begin{remark}\lab{rem:0}
    Gurvits conjectures (in Conjecture 2 of \cite{Gurvits}) that in the setup of Dyer--Gritzmann--Hufnagel from \cite{DGH}, it is impossible to design a Fully-Polynomial-Randomized-Approximation-Scheme (FPRAS) for general mixed volumes.
    \end{remark}
    
   \begin{remark}  \lab{rem:1}
   A randomized polynomial-time algorithm is given in \cite{Gurvits} that produces a vector $\mathbf{\la} = (\la_1, \dots, \la_n)$ such that $$\log \left( \frac{V_K(\la_1, \ldots, \la_n)}{\prod_{i=1}^n \la_i}\right) < \eps + \log \mathrm{Cap}(V),$$ and thereby approximates  \( \log \mathrm{Cap}(V) \) to within an additive error of $\epsilon$, with probability greater than $1 - \de$. 
   \end{remark}
    When each $K_i$ is a convex hull $x_i \subseteq B_\infty(2^L)$ of at most $m_0$ points in $\Z^n$, the algorithm in \cite{Gurvits} runs in time polynomial in \( n \), $m_0$, $L$, $\epsilon^{-1}$ and $\log \de^{-1}$.

 

\begin{table}[h!]
\centering
\renewcommand{\arraystretch}{1.3}
\begin{tabular}{|c|p{10cm}|}
\hline
\textbf{Symbol} & \textbf{Meaning} \\
\hline
$P_1, \dots, P_k$ & Convex polytopes in $\mathbb{R}^n$, each given as the convex hull of at most $m_0$ points in $\mathbb{Z}^n$ \\
\hline
$\alpha = (\alpha_1, \dots, \alpha_k)$ & A tuple of nonnegative integers such that $\sum_i \alpha_i = n$ \\
\hline
$V(P_1^{(\alpha_1)}, \dots, P_k^{(\alpha_k)})$ & Mixed volume corresponding to the multiset $(P_1^{(\alpha_1)}, \dots, P_k^{(\alpha_k)})$ \\
\hline
$V_K(\lambda_1, \dots, \lambda_k)$ & Minkowski polynomial: volume of $\lambda_1 K_1 + \cdots + \lambda_k K_k$ \\
\hline
$\text{Cap}(V)$ & Gurvits' capacity: $\inf_{x > 0} \frac{V_K(x_1,\dots,x_n)}{\prod x_i}$ \\
\hline
$\text{Cap}_\alpha(p)$ & Weighted capacity: $\inf_{x > 0} \frac{p(x)}{x^\alpha}$ \\
\hline
$\widetilde{A}$ & The product $\prod_{i=2}^k \frac{(\alpha_i + 1)^{\alpha_i + 1}}{\alpha_i^{\alpha_i}}$, used in bounding the runtime \\
\hline
$A$ & A refined constant (Definition 2), possibly smaller than $\widetilde{A}$, used in runtime bounds \\
\hline
$B_\infty(2^L)$ & The $\ell_\infty$ ball of radius $2^L$, assumed to contain each $P_i$ \\
\hline
$N(F, P)$ & Normal cone of face $F$ of polytope $P$ \\
\hline
$[F_1, \dots, F_k]$ & Volume of the parallelepiped formed by summing unit cubes in the affine hulls of $F_i$ \\
\hline
$\lambda = (\lambda_1, \dots, \lambda_k)$ & Vector used to evaluate the Minkowski polynomial; chosen to approximate the capacity \\
\hline
$Q_z$ & Polytope of weight vectors  $(w_{ij})_{i, j}$ such that $\sum_{i,j} w_{ij} \lambda_i v_{ij} = z$ and $\sum_{i,j} w_{ij} = 1$ \\
\hline
$\mathbf{w}$ & Vertex of $Q_z$ used to decompose $z$ as a convex combination of vertices from each $\lambda_i P_i$ \\
\hline
$z^{(i)}$ & Weighted point in $P_i$ corresponding to the decomposition of $z$ \\
\hline
$F_i$ & Face of $P_i$ containing $z^{(i)}$ in its relative interior \\
\hline
$W_1(\mu, \nu)$ & Wasserstein-1 distance between probability measures $\mu$ and $\nu$ \\
\hline
$N(p)$ & Normalization operator: converts ordinary generating function $p$ into an exponential generating function \\
\hline
\end{tabular}
\caption{Table of notation}
\end{table}

\section{High-level Road-map}
To orient the reader, we sketch the main ideas before presenting details.

\begin{enumerate}[label=(\arabic*)]
  \item \textbf{Capacity surrogate.}  
        We use a generalisation of Gurvits’ capacity due to \cite{BLP} to arbitrary multiplicities, approximate it in
        $\mathrm{poly}(n,m_0, L,\epsilon^{-1},\log\delta^{-1})$ time, and find scaling factors $\la_1, \dots, \la_k$.
  \item \textbf{Subdivision of the Minkowski sum.}  
        Schneider’s theorem decomposes $\sum_{i}\lambda_i P_i$ into disjoint face-sums
        compatible with the multiplicity vector $\alpha$.
  \item \textbf{Sampling estimator.}  
        We sample
        $z\in\sum\lambda_i P_i$ nearly uniformly, decompose $z$ via an LP vertex,
        and test face dimensions.
         A Chernoff bound over $\texttt{N}=\Theta(A\epsilon^{-2}\log\delta^{-1})$
        samples yields a $(1\pm\epsilon)$ estimator of the mixed volume.
\end{enumerate}

 \section{Bounds from the theory of Lorentzian polynomials}
 
 We have also have bounds depending on $\a$ relating the capacity to the mixed volume, due to Br\"{a}nd\'{e}n, Leake and Pak \cite{BLP}. Before giving the details, we will need a refined notion of capacity, $\operatorname{Cap}_\a$ that is used in the results of $\cite{BLP}$.
 
 \begin{definition}[Capacity {\cite[Def.~5.1]{BLP}}]
Let \(p \in \mathbb{R}_{\ge 0}[x_{1},\dots,x_{k}]\) be a polynomial with non-negative coefficients, and let 
\(\alpha=(\alpha_{1},\dots,\alpha_{k}) \in \mathbb{Z}^{k}_{\ge 0}\).
The \emph{capacity} of \(p\) with respect to the weight vector \(\alpha\) is
\[
  \operatorname{Cap}_{\alpha}(p)
  := \inf_{x>0}\frac{p(x)}{x^{\alpha}}
  \;=\;
  \inf_{x_{1},\dots,x_{k}>0}
        \frac{p(x_{1},\dots,x_{k})}{
              x_{1}^{\alpha_{1}}\cdots x_{k}^{\alpha_{k}} }.
\]
\end{definition}

Let $N$ be the linear operator defined by the condition $N(w^\a) = \frac{w^\a}
{\a!}$. The
normalization operator $N$ turns generating functions into exponential generating functions.
Given a polynomial $p$, if $N(p)$ is Lorentzian, $p$ is called denormalized Lorentzian \cite{BLP}.
By Corollary 3.7 of Br\"{a}nd\'{e}n-Huh \cite{BH},  if $f$ is a Lorentzian polynomial (see Definition 2.1 of \cite{BH} for a definition of this class of polynomials), then it follows that
$N(f)$ is a Lorentzian polynomial.

Thus all Lorentzian polynomials are also denormalized Lorentzian polynomials. 
By Theorem 4.1 of \cite{BH}, we have the following.
\begin{theorem}[Br\"{a}nd\'{e}n-Huh, Thm.~4.1]
The volume polynomial $V_K$ is a Lorentzian polynomial
for any $k$-tuple of convex sets  $K = (K_1,...,K_k)$.
\end{theorem}

 \begin{theorem}[Brändén–Leake–Pak, Thm.~5.10]
Let \(p\in\mathbb{R}_{>0}[x_{1},\dots ,x_{k}]\) be a \emph{denormalised Lorentzian} polynomial of total degree \(n\),
\[
  p(x_{1},\dots ,x_{k}) \;=\;
  \sum_{\mu_{1}+\dots+\mu_{k}=n} p_{\mu}\,x^{\mu}.
\]

For \(1\le i\le k-1\) define
\[
  d_{i} \;:=\;
  \deg_{x_{i}}\!\Bigl(
        \partial_{i+1}^{\alpha_{i+1}}\!\cdots\!
        \partial_{k}^{\alpha_{k}}\,p
      \Bigr)\Big|_{x_{i+1}=\dots =x_{n}=0},
  \qquad
  d_{k}:=\deg_{x_{k}}p.
\]

Let \(\alpha=(\alpha_{1},\dots ,\alpha_{k})\in\mathbb{N}^{k}\) satisfy
\(\alpha_{1}+\dots+\alpha_{k}=n\).
Then
\[
  \frac{p_{\alpha}}{\operatorname{Cap}_{\alpha}(p)}
  \;\ge\;
  \biggl[
    \prod_{i=2}^{k}
      \operatorname{Cap}_{(d_{i}-\alpha_{i},\,\alpha_{i})}
      \!\Bigl(\,\sum_{j=0}^{d_{i}} x^{j}y^{\,d_{i}-j}\Bigr)
  \biggr]^{-1}
  \;\ge\;
  \prod_{i=2}^{k}
  \max\!\Biggl\{
     \frac{\alpha_{i}^{\,\alpha_{i}}}{(\alpha_{i}+1)^{\alpha_{i}+1}},
     \frac{(d_{i}-\alpha_{i})^{\,d_{i}-\alpha_{i}}}
          {(d_{i}-\alpha_{i}+1)^{\,d_{i}-\alpha_{i}+1}}
  \Biggr\}.
\]
\end{theorem}
\begin{definition} 
\lab{def:4}
When $x = (x_1, \dots, x_k)$, for a given $\a$ as above, let the corresponding constant $ \prod_{i=2}^{k}
 \min\!\Biggl\{
     \frac{(\alpha_{i}+1)^{\alpha_{i}+1}}{\alpha_{i}^{\,\alpha_{i}}},
     \frac{(d_{i}-\alpha_{i}+1)^{\,d_{i}-\alpha_{i}+1}}{(d_{i}-\alpha_{i})^{\,d_{i}-\alpha_{i}}}
  \Biggr\}
 $ be denoted by $A$.  
 \end{definition}
 
 It is a property of $V_K(x_1, \dots, x_k)$,  that $d_i \leq \mathrm{deg}_{x_i}V_K = \dim(K_i)$. 
\begin{lemma} 
The quantity $A$ is bounded above by $(\frac{e(n+k-1)}{k-1})^{k-1}$.
\end{lemma}
\begin{proof}
Note that $A$ is bounded above by $\prod_{i=2}^{k}\frac{(\alpha_{i}+1)^{\alpha_{i}+1}}{\alpha_{i}^{\,\alpha_{i}}}$.  By taking a logarithm,  and using the concavity of the logarithm,  we see for $\a_i > 1$,  that $(1 + \a_i^{-1})^{\a_i} \leq e$ for all positive real values of $\a_i$.  Secondly,  $\prod_{i=2}^{k} (\alpha_{i}+1),$ by the A.M.-G.M. inequality,  is bounded above by $\left(\frac{(n+k-1)}{k-1}\right)^{k-1}$. The lemma follows.
\end{proof}

\noindent From Remark~\ref{rem:1}, it follows that it is possible to use the randomized polynomial-time algorithm given in \cite{Gurvits} to produce a vector $\mathbf{\la} = (\la_1, \dots, \la_k) \in \R_+^k$ such that $$\log \text{Cap}_\a(V) \leq \log \left( \frac{V_K(\la_1, \ldots, \la_k)}{\prod_{i=1}^n \la_i^{\a_i}}\right) < \eps + \log \text{Cap}_\a(V),$$ with probability greater than $1 - \de$. Proposition~\ref{prop:6} shows that $\la$ can in fact be chosen to be an integer vector with polynomial bitlength.

\begin{proposition}[Proposition 3.4, \cite{Gurvits}]\lab{prop:gurv}
Consider an  $n$-tuple of convex compact sets 
$K = (K_1, \ldots, K_n)$ with 
$A_i(\mathrm{B}_n(0, 1)) \subset K_i \subset y + n\sqrt{n + 1} \times A_i(\mathrm{B}_n(0, 1))$, 
$1 \le i \le n$, with integer $n \times n$ matrices $A_i$. Let $\langle K \rangle$ denote the bitlength of $(A_1, \dots, A_n)$. Then the minimum in the 
convex optimization problem
\begin{equation}
\log\left( \mathrm{Cap}(V_K) \right) = \inf_{\substack{y_1 + \cdots + y_n = 0}} 
\log\left( V_K(e^{y_1}, \ldots, e^{y_n}) \right).
\tag{2}
\end{equation}
 is attained and is unique. The unique minimizing 
vector $(z_1, \ldots, z_n)$, $\sum_{i=1}^n z_i = 0$, satisfies the following inequalities:
\[
|z_i - z_j| \le O\left(n^{3/2}(\log n + \langle K \rangle)\right), \quad 
\|(z_1, \ldots, z_n)\|_2 \le O\left(n^2(\log n + \langle K \rangle)\right).
\]
In other words, the convex optimization problem \emph{(2)} can be solved on the following ball in $\mathbb{R}^{n-1}$:
\[
\mathrm{Apr}(K) = \left\{ (z_1, \ldots, z_n) : \|(z_1, \ldots, z_n)\|_2 \le O\left(n^2(\log n + \langle K \rangle)\right), \ \sum_{i=1}^n z_i = 0 \right\}.
\]
The following inequality follows from the Lipschitz property:
\[
\left| \log \left( V_K(e^{y_1}, \ldots, e^{y_n}) \right) - 
\log \left( V_K(e^{\ell_1}, \ldots, e^{\ell_n}) \right) \right|
\le O\left(n^3 (\log n + \langle K \rangle)\right), \quad Y, L \in \mathrm{Apr}(K).
\]
\end{proposition}

\begin{proposition} \lab{prop:6}
It is possible to find and output $\la = (\la_1, \dots, \la_k)$ that is a vector with positive integer entries, in randomized polynomial time in $n, m_0, L, A, \epsilon^{-1}$ and $\log \de^{-1}$ such that $$ \log \mathrm{Cap}_\a(V) \leq \log \left( \frac{V_K(\la_1, \ldots, \la_k)}{\prod_{i=1}^k \la_i^{\a_i}}\right) < \eps + \log \mathrm{Cap}_\a(V),$$ with probability greater than $1 - \de$.
\end{proposition}
\begin{proof}
By Proposition 3.2 of \cite{Gurvits}, $\log V_K(e^{y_1}, \dots, e^{y_k})$ is $n$-Lipschitz function of $\R^k$ equipped with the Euclidean norm, when each convex set $K_i \subseteq \R^n$.
Proposition~\ref{prop:6} is now an immediate consequence of Proposition~\ref{prop:gurv} above.
\end{proof}
\begin{notation} Let $\de_1 := \frac{\de 2^{-C nL}}{100 n}.$ and $\de_2 := \frac{\de_1\de^2 2^{-D}}{1000n^2}$ be a small real numbers where $D$ is a large positive integer bounded below by $km_0 + C'nL$ for a suitable positive absolute constant $C'$. We write $C, c, C'$, etc. to denote various positive universal constants whose value may change from one line to the next.
We suppose that $P_i$ is the convex hull of $V_i$ and that no vertex $v_j \in V_i$ is in the convex hull of $V_i \setminus \{v_j\}$, which can be enforced by examining the $V_i$ before the onset of the algorithm, and removing the $v_j$ that violate this condition.
\end{notation}

\section{A subdivision of the Minkowski sum}
\begin{definition}[Normal cone] The normal cone $N (F, P )$ of a face $F$ of a polytope $x$ is the set
of all vectors $v \in \R^n$ such that the supporting hyperplane $H(P, v)$ contains $F$,
i.e.,
$N (F, P ) = \{v \in \R^n|
F \subseteq H(P, v) \cap P \}.$ In other words, $N(F, P)$ is the set of all vectors $v\in \R^n$ such that for every point $y$ in $F$, $\max\limits_{x \in P} \langle x, v\rangle$ equals $\langle y, v \rangle$.
\end{definition}
The following theorem was proved in  \cite{Sch94}, and appears as
Theorem 15.2.3 of \cite{henk}. 
Given a convex polytope $K$, let  $\operatorname{relint} K$ denote its relative interior.
\begin{theorem}[Schneider]\lab{thm:SSF}
Let $P_1, \ldots, P_k \subseteq \mathbb{R}^n$ be polytopes, $k \geq 2$. Let $x^1, \ldots, x^k \in \mathbb{R}^n$ with $x^1 + \cdots + x^k = 0$, $(x^1, \ldots, x^k) \neq (0, \ldots, 0)$, and
\[
\bigcap_{i=1}^k \left( \operatorname{relint} N(F_i, P_i) - x^i \right) = \emptyset
\]
whenever $F_i$ is a face of $P_i$ and $\dim(F_1) + \cdots + \dim(F_r) > n$. Then
\[
\binom{n}{\a_1, \ldots, \a_k} V(P_1^{(\a_1)}, \ldots,  P_k^{(\a_k)}) = \sum_{(F_1, \ldots, F_k)} V(F_1 + \cdots + F_k),
\]
where the summation extends over the $k$-tuples $(F_1, \ldots, F_k)$ of $\a_i$-faces $F_i$ of $P_i$ with $\dim(F_1 + \cdots + F_k) = n$ and $\bigcap_{i=1}^k (N(F_i, P_i) - x^i) \neq \emptyset$.

The choice of the vectors $x^1, \ldots, x^k$ implies that the selected $\a_i$-faces $F_i \subseteq P_i$ of a summand $F_1 + \cdots + F_k$ are contained in complementary subspaces. Hence one may also write
\[
\binom{n}{\a_1, \ldots, \a_k} V(P_1^{(\a_1)}, \ldots,  P_k^{(\a_k)}) = \sum_{(F_1, \ldots, F_k)} [F_1, \ldots, F_k] \cdot V^{\a_1}(F_1) \cdots V^{\a_k}(F_k),
\]
where $[F_1, \ldots, F_k]$ denotes the volume of the parallelepiped that is the sum of unit cubes in the affine hulls of $F_1, \ldots, F_k$.

Finally, we remark that the selected sums of faces in the formula of the theorem form a subdivision of the polytope $P_1 + \cdots + P_k$, i.e.,
\beq \lab{eq:key}
P_1 + \cdots + P_k = \biguplus_{(F_1, \ldots, F_k)} (F_1 + \cdots + F_k).
\eeq
\end{theorem}

\begin{center}
\begin{tikzpicture}[scale=2]

  \coordinate (S1) at (0,0);
  \coordinate (S2) at (1,0);
  \coordinate (S3) at (1,1);
  \coordinate (S4) at (0,1);

  \coordinate (t1) at (0,0);            
  \coordinate (t2) at (0.5,1.2);        
  \coordinate (t3) at (1,0);            

  \coordinate (T1) at ($(S3)+(t1)$);
  \coordinate (T2) at ($(S3)+(t2)$);
  \coordinate (T3) at ($(S3)+(t3)$);

  \coordinate (M1) at ($(S1)+(t1)$);
  \coordinate (M2) at ($(S2)+(t1)$);
  \coordinate (M3) at ($(S2)+(t3)$);
  \coordinate (M4) at ($(S3)+(t3)$);
  \coordinate (M5) at ($(S3)+(t2)$);
  \coordinate (M6) at ($(S4)+(t2)$);
  \coordinate (M7) at ($(S4)+(t1)$);

  \fill[gray!30] (M1) -- (M2) -- (M3) -- (M4) -- (M5) -- (M6) -- (M7) -- cycle;
  \draw[thick] (M1) -- (M2) -- (M3) -- (M4) -- (M5) -- (M6) -- (M7) -- cycle;

  \fill[blue!20] (S1) -- (S2) -- (S3) -- (S4) -- cycle;
  \draw[blue, thick] (S1) -- (S2) -- (S3) -- (S4) -- cycle;
  \node[blue] at (0.5, -0.15) {$P_1$};

  \fill[red!20] (T1) -- (T2) -- (T3) -- cycle;
  \draw[red, thick] (T1) -- (T2) -- (T3) -- cycle;
  \node[red] at ($(T2)+(0.3,0)$) {$P_2$};

\end{tikzpicture}
\end{center}

\begin{figure}[ht]
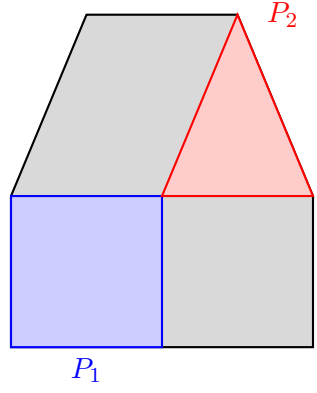

 \caption{The above figure,  illustrates Theorem~\ref{thm:SSF}. Here $k= n = 2$ and $P_1$ and $P_2$ are respectively the blue square  and the red triangle. The origin is $P_1 \cap P_2$ and $P_1 + P_2$ is the polygon above, subdivided into convex polygons of the form $F_1 + F_2$ as stated in Theorem~\ref{thm:SSF}.} The respective Lebesgue measures  of the blue, red and grey portions equal $V(P_1, P_1)$, $V(P_2, P_2)$, and $2V(P_1, P_2)$ respectively.
  \label{fig:example}

\end{figure}

\section{Algorithm}\lab{sec:4}

Given a target vector $z \in \mathbb{Q}^n \cap\{\sum_{i = 1}^k \la_i P_i\}$ and a set of
vertices 
\[
\mathcal{V}\;=\;\bigl\{\la_i v_{ij}\in\mathbb{Z}^n \,: \; i=1,\dots,k,\; v_{ij} \in V_i,\;
|V_i|\le m_0\bigr\},
\]
find a vertex ${\bf w} = \{w_{ij}\}_{i \in [k], v_{ij} \in V_i}$ 
that maximizes  $$\sum_{i=1}^k \langle x^i, \sum_{v_{ij} \in V_i} w_{ij}\la_i v_{ij}\rangle$$ in the polytope $Q_z$ of non-negative weights $w_{ij}$ such that  
\begin{align}\tag{LP}
  &\sum_{i=1}^k \sum_{v_{ij}\in V_i} w_{ij}\la_i v_{ij} = z, \label{eq:obj}\\
  &\sum_{i=1}^k \sum_{v_{ij}\in V_i} w_{ij}         = 1, \nonumber \\
  &w_{ij} \ge 0 \quad \forall\,i\in\{1,\dots,k\},\; v_{ij} \in V_i, \nonumber 
\end{align}
Let $z^{(i)} = \sum_j w_{ij} v_i \in \Q^n.$ 

\begin{algorithm}
\caption{Sampling and Facet Dimension Test}
\begin{algorithmic}[1]
\State Set \texttt{N} to $\Omega(A\eps^{-2} \log \de^{-1})$.
\State Compute $\{\la_1, \dots, \la_k\}$ using Proposition~\ref{prop:6}.
\State Compute an estimate $\widehat{V}_K(\la_1, \ldots, \la_k)$ that with probability greater than $ 1 - \frac{\de}{50}$ is a multiplicative $ 1 \pm \eps$ estimate of $V_K(\la_1, \ldots, \la_k)$.
\State Initialize \texttt{T} $\gets 0$.
\For{$\texttt{count} = 1$ to \texttt{N}}
    \State Sample $z $ from a measure $\mu$ supported on $2^{- D_2} \Z^n$ that is $\de_2$-close in $W_1$ to the uniform measure $\nu$ on $\lambda_1 P_1 + \cdots + \lambda_k P_k$, where $D_2 := C \lceil \log \de_2^{-1}\rceil.$
    \State For $i \in [k]$ and $j \in [|V_i|]$, set $w_{ij} \in \Q$ so that the matrix $\mathbf{w} = (w_{ij})_{i, j}$ corresponds to a vertex of $Q_z$ that maximizes the objective $\sum_{i=1}^k \langle x^i, \sum_{v_{ij} \in V_i} w_{ij}\la_i v_{ij}\rangle.$
    \State  For $i \in [k]$,  set $z^{(i)} = \sum_j w_{ij} v_{ij} \in \Q^n.$
      \For{$i = 1$ to $k$}
        \State Find the  face $\hat{F}_i \subseteq P_i$ such that $z^{(i)} \in  \operatorname{relint} \hat{F}_i$ (see Lemma~\ref{lem:relint} for details). 
        \State Set $n_i \gets \dim(\hat{F}_i)$.
    \EndFor
    \If{$\forall i \in [k], n_i = \a_i$}
        \State \texttt{T} $\gets$ \texttt{T} $+ 1$
    \EndIf
\EndFor
\State Output $\frac{\texttt{T}}{\texttt{N}} \left( \frac{\widehat{V}_K(\la_1, \ldots, \la_k)}{\prod_{i=1}^k \la_i^{\a_i}}\right) $.
\end{algorithmic}
\end{algorithm}

\begin{remark} $\dim F_i$ is determined with high probability as follows.
Given fixed $i' \in [k]$, and  $z^{(i')} = \sum_j w_{i'j} v_{i'j}  \in P_{i'}$, we know that $F_{i'}$ defined by $z^{(i')} \in  \operatorname{relint} \hat{F}_{i'}$ contains the set of points  $$V_{i'}^{z} := \{v_{i'j} \in V_{i'} \text{ such that }\, w_{i'j} \neq 0\}.$$ However, we have the following Lemma~\ref{lem:relint}. \end{remark}
\begin{lemma}\lab{lem:relint} With probability at least $1 - \frac{\de}{{\texttt{N}}(m_0 n)^C}$, $$\dim F_{i'} = |V_{i'}^z| -1.$$ \end{lemma}
\begin{proof} This is because $\mathrm{conv}(V_{i'}^z) + \sum_{j \neq i'} F_i$ is a polytope contained in $Q_z$ of positive codimension, 
and $Q_z$ (a polytope whose diameter is bounded above by $n 2^{CL}$) contains a full dimensional ball of radius at least $2^{-CnL},$ and $\log D_2$ is larger than $CnL$ by a universal multiplicative constant.
\end{proof}

The \textbf{Wasserstein distance} \( W_1 \) between two Borel probability measures \( \mu \) and \( \nu \) on \( \mathbb{R}^n \) is defined as:
\[
W_1(\mu, \nu) := \inf_{\pi \in \Pi(\mu, \nu)} \int_{\mathbb{R}^n \times \mathbb{R}^n} \|x - y\| \, d\pi(x, y),
\]
where:
\begin{itemize}
    \item \( \|x - y\| \) denotes the Euclidean norm on \( \mathbb{R}^n \),
    \item \( \Pi(\mu, \nu) \) is the set of \emph{couplings} of \( \mu \) and \( \nu \), i.e., probability measures \( \pi \) on \( \mathbb{R}^n \times \mathbb{R}^n \) such that
    \[
    \pi(A \times \mathbb{R}^n) = \mu(A), \quad \pi(\mathbb{R}^n \times B) = \nu(B),
    \]
    for all Borel sets \( A, B \subseteq \mathbb{R}^n \).
\end{itemize}

\begin{remark}\lab{rem:2} We will further assume that for each $i \in [k]$, the $x^{(i)}$ does not belong to the {\bf boundary} $\partial N(F_i, P_i)$ of an outer normal cone for any face $F_i \subseteq P_i$. This assumption holds true with high probability if $x = (x^{(1)}, \dots, x^{(k)})$ is chosen uniformly at random from the set of rational points in any bounded convex set $K$ containing the unit ball $(B_n(0, 1))^{\times k} \cap \{(x^{(1)}, \dots, x^{(k)})| \sum_i x^{(i)} = 0\}$, expressible as a ratio of integers with denominator $2^{D_2}$ such that $\sum_i x^{(i)} = 0.$
\end{remark}

\section{Main theorem}

\begin{theorem}\lab{thm:main}
 $\mathbf{Algorithm \, 1}$  produces an estimate of $V(P_1^{(\a_1)}, \ldots, P_k^{(\a_k)})$ that is within a multiplicative $1 \pm \epsilon$ factor of the true mixed volume with a probability greater than $1 - \de.$  When each $P_i \subseteq B_\infty(2^L)$, the time complexity is bounded above by a polynomial in $n, m_0, L, A, \epsilon^{-1}$ and $\log \de^{-1}$. 
 \end{theorem}
\begin{proof}
The polytope $Q_z$ is nonempty with probability $1$, when $z$ is chosen uniformly at random from $\sum_i \la_i P_i$ by Theorem~\ref{thm:SSF}.

Let $z \in \la_1 P_1 + \dots + \la_k P_k$ be sampled from a measure $\de_2$-close in $W_1$ to the uniform measure on $\la_1 P_1 + \dots + \la_k P_k$.
By (\ref{eq:key}), $$\la_1 P_1 + \cdots + \la_k P_k = \biguplus_{(\la_1 F_1, \ldots, \la_k F_k)} (\la_1 F_1 + \cdots + \la_k F_k),$$ where the summation extends over the $k$-tuples $(F_1, \ldots, F_k)$ of $\a_i$-faces $F_i$ of $P_i$ with $\dim(F_1 + \cdots + F_k) = n$ and $\bigcap_{i=1}^k (\operatorname{relint} N(\la_i F_i, \la_i P_i) - x^i_\la) \neq \emptyset$, where the $x^i_\la$ correspond to the points $x^i$ that appear in Theorem~\ref{thm:SSF}, and we  just choose $x_\la^i = x^i$, because all the conditions of Theorem~\ref{thm:SSF} are satisfied with this choice.

Let $\check{z}$ be a point sampled from the uniform measure  on $\la_1 P_1 + \dots + \la_k P_k$. Then, by choosing an optimal coupling $\pi$ with respect to $W_1$, we see that there exists an $\R^n\oplus \R^n$ valued random variable  $(\check{x}, \eta)$   such that $\check{x} + \eta$ has the same distribution as $x$ in the statement of Algorithm 1, where $\E[|\eta|] \leq \de_2 = \frac{\de_1 \de^2 2^{-D}}{1000 n^2}$.
 Let $P_{\check{z}} := \la_1 F_1 + \cdots + \la_k F_k$, where $\dim(F_1 + \cdots + F_k) = n$ and $\bigcap_{i=1}^k (\operatorname{relint} N(F_i, P_i) - x^i) \neq \emptyset$ be a (closed) polytope that contains $\check{z}$, and let $P_{\check{z}}^\circ$ denote its interior. Such a polytope exists and is unique with probability $1$ over the choice of $\check{z}$ by Theorem~\ref{thm:SSF}.

By Markov's inequality, \beq\lab{eq:2} \pi\left(|\eta| < \frac{\de_2}{10\de n}\right) > 1 - \frac{\de}{100}.\eeq 

\begin{lemma} \lab{lem:9} We have that \beq\lab{eq:3}\pi\left(B(\check{z}, \frac{\de_2}{10\de n}) \subseteq P_{\check{z}}^\circ\right) > 1 - \frac{\de}{100}.\eeq
\end{lemma}
\begin{proof}
This follows from the fact that that $\de_1< \frac{\de 2^{-CnL}}{100 n}$ and that each facet $\la_i F_i$ of the lattice polytope $\la_i P_i$ contains a $\dim(F_i)$ dimensional simplex with integral vertices (each of whose norm in $\ell_\infty^n$ is less than $2^L$,) and hence a $\dim(F_i)$ dimensional ball of radius $10 \de_1$.  

\end{proof}
By Lemma~\ref{lem:9} and (\ref{eq:2}) above, \beq \lab{eq:4}  \pi\left(z \in P_{\check{z}}^\circ\right) > 1 - \frac{\de}{50}. \eeq
We will denote the event that $z \in P_{\check{z}}^\circ$ by $E$.

\begin{lemma}\lab{lem:11}
 Given the decomposition $z = \sum_i z^{(i)}$ obtained from a vertex $\mathbf{w}$ of $ Q_z$  that maximizes the objective $\sum_{i=1}^k \langle x^i, \sum_{v_{ij} \in V_i} w_{ij}\la_i v_{ij}\rangle$
 as in Algorithm 1,  if $\hat{F}_i$ is the (unique) face such that $\hat{F}_i \subseteq P_i$ and $z^{(i)} \in  \operatorname{relint} \la_i \hat{F}_i$, then $$\pi\left(\forall i \in [k], \hat{F}_i = F_i \right) > 1- \frac{\de}{10}.$$ 
\end{lemma}
\begin{proof}
By (\ref{eq:key}), Lemma~\ref{lem:9} and (\ref{eq:4}), it suffices to prove that with probability greater than $1 - \frac{\de}{10}$ we have both that $\dim(\hat{F_1} + \cdots + \hat{F}_k) = n$ and that $\bigcap_{i=1}^k (\operatorname{relint}  N(\la_i \hat{F}_i, \la_i P_i) - x^i_\la) \neq \emptyset$.
 We first prove 
   \begin{claim} \lab{cl:1} $$\pi\left(\bigcap_{i=1}^k (\operatorname{relint}  N(\la_i \hat{F}_i, \la_i P_i) - x^i_\la) \neq \emptyset \bigg|E\right) = 1.$$ \end{claim}
 \begin{proof}
 Conditional on $E$, the primal LP labelled (LP) is bounded and feasible. This LP can be rewritten as  follows. Find, for $i \in [k]$, $z^{(i)}$ that:
 \begin{equation}
 \begin{aligned}
& \text{maximize} && \sum_{i=1}^k \langle x^{i}, z^{(i)}  - \frac{z}{k}\rangle \\
& \text{subject to} && z^{(i)}- \frac{z}{k} \in \la_i P_i - \frac{z}{k} \quad \text{for all } i = 1, \dots, k, \,\text{and}\\
& && \sum_{i=1}^k (z^{(i)} - \frac{z}{k})= 0.
\end{aligned}
\tag{LP1}
\end{equation}
The state of affairs with (LP1) can be summarized geometrically as follows.
 Consider the direct sum polytope $\bigoplus_{i=1}^k  (\la_i P_i - \frac{z}{k})$ (that is a subset of $\R^{nk}$) intersected with the subspace $S$ of all vectors of the form $(v^{(1)}, \dots, v^{(k)})$ where $\sum_i v^{(i)} = 0,$ and call the resulting polytope $T$. Now we optimize on this polytope $T$  in the direction of $(x^{1}, \dots, x^{k})$. 
 It follows from KKT conditions that at an optimal point $\mathbf{w}'$, $(x^1, \dots, x^k)$ can be expressed as the sum of two vectors, one in the complementary subspace $S^\perp$ orthogonal to $S$, and the other in the outer normal cone at ${\mathbf w}'$ to $\bigoplus_{i=1}^k  (\la_i P_i - \frac{z}{k})$. Now $S^\perp$ consists precisely of those vectors that have the form $(v, \dots, v)$ for some $v \in \R^n$; that this set is contained in $S^\perp$ is clear, and dimensional considerations show that $S^\perp$ can contain no other vector. But the outer normal cone at $\mathbf{w}'$ to $\bigoplus_{i=1}^k  (\la_i P_i - \frac{z}{k})$ is precisely $$\bigoplus_{i=1}^k N(\la_i \hat{F}_i, \la_i P_i).$$ We have the following:
 $\bigoplus_{i=1}^k N(\la_i \hat{F}_i, \la_i P_i)$ contains a vector of the form $(x^1, \dots, x^k) - (v, \dots, v)$.
 And so,  in view of Remark~\ref{rem:2},
 $$\pi\left(\bigcap_{i=1}^k (\operatorname{relint}  N(\la_i \hat{F}_i, \la_i P_i) - x^i) \ni \{- v\}|E\right) = 1,$$ for some vector $v \in \R^n$. This proves Claim~\ref{cl:1}. 
 
 \end{proof}

 \begin{claim}
 $$\pi\left(\dim(\hat{F_1} + \cdots + \hat{F}_k) > n|E\right) = 0.$$
 \end{claim}
 \begin{proof}
 We know that $\pi(\check{z} \in P_{\check{z}}^\circ) = 1$. The conditional probability of the event $z \in P_{\check{z}}^\circ$ given $E$ is also $1$. By the choice of $x_\la^{(i)}$ and the condition that  \[
\bigcap_{i=1}^k \left( \operatorname{relint} N(\la_i \tilde{F}_i, \la_i P_i) - x^i \right) = \emptyset
\]
whenever $\tilde{F}_i$ is a face of $P_i$ and $\dim(\tilde{F}_1) + \cdots + \dim(\tilde{F}_k) > n$ in Theorem~\ref{thm:SSF}, it follows from Claim~\ref{cl:1} that $$\pi\left(\dim(\hat{F}_1 + \cdots + \hat{F}_k) > n|E\right) = 0.$$
\end{proof}
To prove Lemma~\ref{lem:11}, it now suffices by (\ref{eq:4}) to prove the following.  

\begin{claim}
$$\pi\left(\dim(\hat{F}_1 + \cdots + \hat{F}_k) < n|E\right) < \frac{\de}{50}.$$
\end{claim}
\begin{proof}
The total number of subspaces of dimension less than $n$ of the form $$\mathrm{span}(\hat{F}_1 + \cdots + \hat{F}_k)$$ is bounded above by 
$2^{km_0}$. K. Ball in \cite{Ball} proved that the volume of a section of a unit $\ell_\infty$ ball of dimension $n$ is bounded above by $2^n \sqrt{2}$. Thus, the total volume of slabs that are neighborhoods of such subspaces, of thickness $2^{-D}$, contained inside an $\ell_\infty^n$ ball of radius $2^{CnL}$ is bounded above by $\sqrt{2} \de (2^{-D}) (2^{km_0}) (2^{CnL}).$  When $D$ is bounded below by $k m_0 + C' nL$ for an appropriately large positive absolute constant $C'$,  this is indeed less than $$\frac{\de}{50} \vol(\la_1 P_1 + \dots + \la_k  P_k)$$ as needed. This proves the claim.
\end{proof}
Lemma~\ref{lem:11} now follows.
 \end{proof}

We now proceed to analyze the time complexity of Algorithm 1.
By the description in Section 3.5 of \cite{Gurvits}, the complexity of estimating $\widehat{V}_K(\la_1, \ldots, \la_k)$ is bounded above by 
$$O(n^{10}(\log(n)+ \langle K \rangle)^2 \log(\de^{-1})),$$ (based on sampling and volume computation algorithms introduced in \cite{DFK}; for the best known bounds, see \cite{LV}). The results in \cite{LV} shows that the complexity of sampling $\mathbf{w}$ is bounded above by an expression of the same form if one uses a real number model. To sample from lattice points with a small $W_1$-distance, it suffices to note that a total variation distance of less than $\frac{\de 2^{- CL}}{k \sqrt{n}}$ implies that the $W_1$ distance is less than $\de$ because the $\ell_2$ diameter of $\la_1 P_1 + \dots + \la_k P_k$ is less than $2^{CL}k \sqrt{n}$. Finally, the complexity of finding a vertex ${\mathbf w}$ of $Q_z$ is polynomial in the input parameters if one uses the ellipsoid method \cite{khach, gls}. It is known (see for example \cite{Karm}) that the difference in the objective value of an optimal vertex, and a vertex with the second best objective, is bounded below by $2^{- poly(nLD_2)}$. It is possible to find a point with the optimal objective using linear programming followed by linear algebra in polynomial time. The set of points in $Q_z$, where the optimal value of the objective is achieved is a face of $Q_z$. A vertex of this face can be obtained in polynomial time in the input parameters using linear programming.
\end{proof}
\section{Conclusion}

Let $(P_1, \dots, P_k)$ be a $k$-tuple of convex polytopes, each of which is defined as the convex hull of at most $m_0$ points in $\Z^n$.  

 $\mathbf{Algorithm \, 1}$ of Section~\ref{sec:4} produces an estimate of $V(P_1^{(\a_1)}, \ldots, P_k^{(\a_k)})$ that is within a multiplicative $1 \pm \epsilon$ factor of the true mixed volume with a probability greater than $1 - \de.$  Let the constant $ \prod_{i=2}^{k} \frac{(\alpha_{i}+1)^{\alpha_{i}+1}}{\alpha_{i}^{\,\alpha_{i}}}$ be denoted by $\tilde{A}$.  When each $P_i \subseteq B_\infty(2^L)$, the time complexity is bounded above by a polynomial in $n, m_0, L, \tilde{A}, \epsilon^{-1}$ and $\log \de^{-1}$. 

 In fact, in Theorem~\ref{thm:main}, we prove a bound on the complexity that is polynomial in $n, m_0, L, A, \epsilon^{-1}$ and $\log \de^{-1}$ for a more refined quantity $A$ (see Definition~\ref{def:4}) that is less or equal to $\tilde{A}$.

\bibliographystyle{plain}
\bibliography{references} 

\end{document}

%% file: macros-2.tex
\newtheorem{theorem}{Theorem}
\newtheorem{lemma}{Lemma}

\newtheorem{remark}{Remark}

\newtheorem{definition}{Definition}
\newtheorem{notation}{Notation}
\newtheorem{proposition}{Proposition}
\newtheorem{claim}{Claim}

\newcommand{\E}{\ensuremath{\mathbb E}}
\newcommand{\R}{\ensuremath{\mathbb R}}

\newcommand{\lab}{\label}    \def\a{{\mathbf{\alpha}}} \def\de{{\mathbf{\delta}}}   
  \def\beq{\begin{eqnarray}} \def\eeq{\end{eqnarray}} \def\ben{\begin{enumerate}}
\def\een{\end{enumerate}} \def\bit{\begin{itemize}}
\def\bel{\begin{lemma}}
\def\eel{\end{lemma}}
\def\eit{\end{itemize}} \def\beqs{\begin{eqnarray*}} \def\eeqs{\end{eqnarray*}} \def\bel{\begin{lemma}} \def\eel{\end{lemma}}
 \newcommand{\Z}{\mathbb{Z}} \newcommand{\Q}{\mathbb{Q}}

  \newcommand{\la}{\lambda}  
  \def\eps{{\epsilon}}   
\def\vol{\mathrm{vol}}